\documentclass[submission,copyright,creativecommons]{eptcs}
\providecommand{\event}{SOS 2010} 
\usepackage{breakurl}
\usepackage{amssymb}
\usepackage[all]{xy}
\usepackage{url}
\usepackage{pst-all}
\usepackage{wasysym} 

\newcommand{\<}{\sqsubseteq}

\newcommand{\lra}{\longrightarrow}

\newcommand{\s}{\lesssim}

\newcommand{\Conformance}[2]{#1\lesssim_{CS}#2}
\newcommand{\NConformance}[2]{#1\not\lesssim_{CS}#2}
\newcommand{\PreConformance}[2]{#1\lesssim^p_{CS}#2}

\newcommand{\CoCo}[2]{#1\lesssim_{CC}#2}

\newcommand{\SimiCC}{\lesssim_{CC}}

\newcommand{\da}{\downarrow}

\newenvironment{proof}{\par\addvspace{\bigskipamount}
\noindent\textit{\textbf{Proof.}}\ }{\par\addvspace{\bigskipamount}}

\newcommand{\qed}{~\Square}

\title{Equational Characterization of Covariant-Contravariant Simulation and Conformance Simulation Semantics}
\author{Ignacio F\'abregas \quad\qquad David de Frutos Escrig\quad\qquad Miguel Palomino
\institute{Universidad Complutense de Madrid\\ Madrid, Spain} \institute{Departamento de Sistemas Inform\'aticos y Computaci\'on
\thanks{Research supported by the Spanish projects DESAFIOS10 TIN2009-14599-C03-01, TESIS TIN2009-14321-C02-01 and PROMETIDOS
S2009/TIC-1465. The second author worked in this paper during a visit to Reykjavik University sponsored through a grant by the ABEL Extraordinary Chair.}\\} \email{fabregas@fdi.ucm.es \quad\qquad defrutos@sip.ucm.es \quad\qquad miguelpt@sip.ucm.es} }
\def\titlerunning{Equational Characterization of Covariant-Contravariant and Conformance Simulation Semantics}
\def\authorrunning{I. F\'abregas, D. de Frutos Escrig \& M. Palomino}

\begin{document}
\newtheorem{definition}{Definition}
\newtheorem{theorem}{Theorem}
\newtheorem{proposition}{Proposition}
\newtheorem{lemma}{Lemma}
\newtheorem{example}{Example}
\newtheorem{remark}{Remark}

\maketitle

\begin{abstract}
Covariant-contravariant simulation and conformance simulation
generalize plain simulation and try to capture the fact that it is
not always the case that ``the larger the number of behaviors, the
better''. We have previously studied their logical characterizations
and in this paper we present the axiomatizations of the preorders
defined by the new simulation relations and their induced
equivalences. The interest of our results lies in the fact
that the axiomatizations help us to  know the new
simulations better, understanding in particular the role of the
contravariant characteristics and their interplay with the
covariant ones; moreover, the axiomatizations provide us with a
powerful tool to (algebraically) prove results of the
corresponding semantics. But we also consider our results
interesting from a metatheoretical point of view: the fact that
the covariant-contravariant simulation equivalence is indeed
ground axiomatizable when there is no action that exhibits both a
covariant and a contravariant behaviour, but becomes
non-axiomatizable whenever we have together actions of that kind
and either covariant or contravariant actions, offers us a new
subtle example of the narrow border separating axiomatizable and
non-axiomatizable semantics. We expect that by studying these
examples we will be able to develop a general theory separating
axiomatizable and non-axiomatizable semantics.
\end{abstract}

\section{Introduction and some related work}

Simulations are a very natural way to compare systems defined by labeled transition systems or other related mechanisms based on describing the behavior of states by means of the actions they can execute \cite{Park81}. They aim at comparing processes based on the simple premise ``you are better if you can do as much as me, and perhaps some other new things''. This assumes that all the executable actions are controlled by the user (no difference between input and output actions) and does not take into account that whenever the system has several possibilities for the execution of an action it will choose in an unpredictable internal way, so that more possibilities means less control.

In order to cope with these limitations one should consider adequate versions of simulation where the characteristics of actions and the idea of preferring processes that are less non-deterministic are taken into account. This leads to two new notions of simulation:
covariant-contravariant simulation and conformance simulation that we roughly sketched in \cite{DeFrutosEtAl07} and presented in detail in \cite{FabregasFP09}, where we proved that they can be presented as particular instances of the general notion of categorical simulation developed by Hughes and Jacobs \cite{HughesJacobs04}.

Certainly, the distinction between input and output actions or similar classifications is not meant to be new at all and, for instance, they were present in modal transition systems as early as the end of the eighties. They also play a central role in I/O-automata \cite{Lynch88} and more recently appear as component of several works on interface automata \cite{AlfaroH01,LarsenNW06}, where one finds the covariant-contravariant distinction when the guarantees of the specification can only be assumed if the conditions of the specification are satisfied.

Concerning conformance simulation, the first related references are also quite old \cite{Leduc92,Tretmans96}, corresponding to the notion of conformance testing, which is close to failure semantics \cite{BrookesR84}. However, it is a bit surprising that in both cases we lack a basic theory where these notions are presented in a simplified scenario, stressing their main characteristics and properties. We think that the theory of semantics for processes, and particularly the simulation semantics, is a perfect field in which to develop that basic theory. This has been already proved in \cite{FabregasFP09}, where our new simulation semantics were shown to be categorical simulations, thus inheriting all their good properties for free.

In \cite{FabregasFP10} we have also briefly presented the logical characterizations of the two semantics. Now that we already know quite well the behaviour of the two new notions of simulation we can give their algebraic presentation. By the way, although in our previous works on the unified study of process semantics the (classical) covariant character of all the actions had several important consequences, mainly represented by the extremely simple and easy to apply basic axiom for simulation $(\mathsf{S})\; x\<x+y$ (or equivalently, just $0\<y$), we have been able to borrow from \cite{Frutos-EscrigGP09,AcetoEtAl07,Frutos-EscrigGP09b} several ideas about the axiomatization of process semantics that, although not directly applicable due to the special characteristics of the new semantics, can be adequately adapted.

However, not all of the simple and nice results for the algebraic theory of plain (covariant) simulation can be extended to the general covariant-contravariant case. In particular, in order to obtain the maximal genericity, when we defined covariant-contravariant simulations in \cite{FabregasFP09} we admitted not only both covariant and contravariant actions, but also other actions with a bivariant nature. This decision was taken because when presenting a general theory of categorical simulations in \cite{HughesJacobs04}, J.~Hughes and B.~Jacobs already noticed that bisimulation was a particular (in fact, trivial) example of simulation semantics. It was also clear that inverse simulation (namely, contravariant simulation) was also another example, and then we were able to prove that our general covariant-contravariant simulation was another categorical simulation that smoothly combines bisimulation, plain (covariant) simulation and inverse (contravariant) simulation.

Obviously, plain bisimulation has a simple axiomatization, as is the case for plain simulation; we will see in this paper that the preorder defined by our covariant-contravariant simulation can also be finitely axiomatized. When we considered the induced
equivalence, we found indeed a finite axiomatization for the case in which there are no bivariant actions (actions that can be considered as both input and output) in our alphabet. The axiomatization and its completeness proof were obtained by adapting the general techniques in \cite{Frutos-EscrigGP09b,Frutos-EscrigGP09} for the covariant case to our more general covariant-contravariant scenario. However, as soon as a single bivariant action is introduced, and at least one non-bivariant one is also present, then the equational theory of covariant-contravariant simulation equivalence becomes non-finitely axiomatizable, and in fact the proof of this result is extraordinarily simple.

Even if this is a negative result, we think that it will contribute to enlight the narrow border separating axiomatizable and non-axiomatizable process theories, which we expect to continue exploring in the future.

There is a large collection of recent papers where notions close to those
studied here are either developed or applied; a detailed comparison will appear
elsewhere.
However, we insist on the fact that we were not able to find a basic study
where the main results on process theory had been extended to a framework
containing any contravariant characteristics, although it is true that some
small contributions along this direction can be found in some of these papers.
We plan to develop a thorough compilation of the works on this topic by
isolating the places where our foundational study could help to understand the
different developments, as well as looking for applications and new enhancements
to our theory that could be of use to relate all the disconnected work on the
area.
In turn, we hope that this will also provide us with some intuition to
understand those results and produce new formal techniques to obtain proofs
of those, or other interesting results in the area.
So, simply to give a hint, a sample of those works would include
\cite{20years,BenesKLS09,LarsenT88,RacletetAl09}.

\section{Preliminaries}

In this section we summarize some definitions and concepts from \cite{Cirstea06,FabregasFP09} and introduce the notation we are going to use.

Let us recall our two new simulation notions:

\begin{definition}
Given $P=(P,A,\rightarrow_P)$ and $Q=(Q,A,\rightarrow_Q)$, two labeled transition systems (LTS) for the alphabet $A$, and $\{A^r,A^l, A^{\mathit{bi}}\}$ a partition of this alphabet, a \textbf{$(A^r,A^l)$-simulation} (or just a covariant-contravariant simulation) between them is a relation $S\subseteq P\times Q$ such that for every $pSq$ we have:
\begin{itemize}
\item For all $a\in A^r\cup A^{\mathit{bi}}$ and all $p\stackrel{a}{\lra}p'$ there exists $q\stackrel{a}{\lra}q'$ with $p'Sq'$.

\item For all $a\in A^l\cup A^{\mathit{bi}}$, and all $q\stackrel{a}{\lra}q'$ there exists $p\stackrel{a}{\lra}p'$ with $p'Sq'$.
\end{itemize}
We will write $\CoCo{p}{q}$ if there exists a covariant-contravariant simulation $S$ such that $pSq$.
\end{definition}
This definition combines the requirements of plain simulation, for some of the actions, with those of plain ``anti-simulation'', for some of the
remaining actions, imposing both on so-called bivariant actions.

\begin{definition}
Given $P=(P,A,\rightarrow_P)$ and $Q=(Q,A,\rightarrow_Q)$ two labeled transition systems for the alphabet $A$, a \textbf{conformance simulation}
between them is a relation $R\subseteq P\times Q$ such that whenever $pRq$, then:
\begin{itemize}
\item For all $a\in A$, if $p\stackrel{a}{\lra}$, then $q\stackrel{a}{\lra}$ (this means, using the usual notation for process
algebras, that $I(p)\subseteq I(q)$).

\item For all $a\in A$ such that $q\stackrel{a} {\lra}q'$ and $p\stackrel{a}{\lra}$, there exists some $p'$ with $p\stackrel{a}{\lra}p'$ and
$p'R q'$.
\end{itemize}
We will write $\Conformance{p}{q}$ if there exists a conformance simulation $R$ such that $pRq$.
\end{definition}
The first clause of the definition guarantees that $Q$ has at least all the behaviors of
$P$, allowing to ``improve'' a process by extending the set of actions it offers, whereas the second clause establishes that a process can be
``improved'' by reducing the nondeterminism in it.

Let us recall that the set \textit{BCCSP(A)} of basic processes for the alphabet $A$ is defined by the $BNF$-grammar
\[p::=0\mid ap\mid p+p\]
where $a\in A$.
The operational semantics for BCCSP terms is defined by
\[%
\begin{array}{ccccc}
 ap\stackrel{a}{\lra}p &\mbox{\qquad} & \frac{\displaystyle p\stackrel{a}{\lra}p'}{\displaystyle
 p+q\stackrel{a}{\lra}p'} &\mbox{\qquad} & \frac{\displaystyle q\stackrel{a}{\lra}q'}{\displaystyle
  p+q\stackrel{a}{\lra}q'}\\
\end{array}
\]
With these operators we can only define finite processes; however, it is well known that these operators capture the essence of any transition system, which can
be defined by a system of equations specifying the behavior of each state.
(The axioms for recursive processes, other interesting extensions
including the communication operators, and possibly some others, are left for
future work.)

\section{Axiomatization of the new simulation preorders}

In this section we present a finite axiomatization of the two preorders for basic finite processes induced by our new kinds of simulation.

\subsection{Covariant-contravariant semantics}
We consider a partition $\{A^r,A^l,A^\mathit{bi}\}$ of the alphabet $A$, with actions that have either a covariant nature, or contravariant, or both at the same time.
Contravariant simulation $\lesssim_S^{-1}$ is just the inverse of plain simulation and therefore can be trivially axiomatized by inverting the axiom for plain
simulation
\begin{itemize}
\item[] $(\mathsf{S})\quad {x}\<{x+y}$,
\end{itemize}

\noindent thus obtaining

\begin{itemize}
\item[] $(\mathsf{S^{-1}})\quad {x+y}\<{x}$.
\end{itemize}

In order to produce an axiomatization of
covariant-contravariant simulation we need to combine in an adequate way these two axioms, by constraining each of them to the case in which the
added process $y$ only offers actions with the corresponding covariant or contravariant character. Hence we obtain:
\begin{itemize}
\item[ ] $(\mathsf{S}^r)\quad {I(y)\subseteq A^r}\Longrightarrow {{x}\<{x+y}}$.

\item[ ] $(\mathsf{S}^{-1,l})\quad {I(y)\subseteq A^l}\Longrightarrow {{x+y}\<{x}}$.
\end{itemize}
We can omit the conditions in these two axioms by considering two generic actions $a_r\in A^r$ and $a_l\in A^l$:
\begin{itemize}
\item[ ] $(\mathsf{S}^r_p)\quad {{x}\<{x+a_r y}}$.

\item[ ] $(\mathsf{S}^{l}_p)\quad {{x+a_l y}\<{x}}$.
\end{itemize}

Note that actions in $A^\mathit{bi}$ do not appear in the axioms above, although
they could be included in the processes instantiating the variables
$x$ and $y$.
This is an immediate consequence of the fact that their behavior
corresponds to that governed by bisimulation, so that we need not add any new axiom to those capturing the bisimilarity relation:
\begin{itemize}
\item[] $(\mathsf{B}_1)\quad {{x+y}={y+x}}$.

\item[] $(\mathsf{B}_2)\quad {{(x+y)+z}={x+(y+z)}}$.

\item[] $(\mathsf{B}_3)\quad {{x+x}={x}}$.

\item[] $(\mathsf{B}_4)\quad {x+0}={x}$.
\end{itemize}
We will use these axioms implicitly in the remainder of this paper.

\begin{proposition} The $(A^r,A^l)$-simulation preorder can be axiomatically defined by means of the set of axioms  $\{\mathsf{B}_1,\mathsf{B}_2,\mathsf{B}_3,\mathsf{B}_4,\mathsf{S}^r_p, \mathsf{S}^{l}_p\}$.
\end{proposition}

\begin{proof}
First we prove that the axioms $(\mathsf{S}^r_p)$ and $(\mathsf{S}^{l}_p)$ are sound for the $(A^r,A^l)$-similarity relation $\SimiCC$. Indeed: 

\begin{itemize}
\item For all $a\in A^r\cup A^\mathit{bi}$, if $x\stackrel{a}{\lra}x'$ then $x+a_r y\stackrel{a}{\lra}x'$ and $x'\SimiCC x'$.

\item For all $a\in A^l\cup A^\mathit{bi}$, if $x+a_r y\stackrel{a}{\lra}x'$, then  $x\stackrel{a}{\lra}x'$ and $x'\SimiCC x'$. Note that $a\neq
a_r$ since $A^r\cap (A^l\cup A^\mathit{bi})=\emptyset$.

\item For all $a\in A^r\cup A^\mathit{bi}$, if $x+a_l y\stackrel{a}{\lra}x'$ then $x\stackrel{a}{\lra}x'$ and $x'\SimiCC x'$ as above, because
$a\neq a_l$ again.

\item For all $a\in A^l\cup A^\mathit{bi}$, if $x\stackrel{a}{\lra}x'$, then $x+a_l y\stackrel{a}{\lra}x'$ and $x'\SimiCC x'$.
\end{itemize}

\noindent To prove completeness we consider $p\SimiCC q$ and reason by structural induction on $p$.

\begin{itemize}
\item If $p$ is $0$ then $I(q)\subseteq A^r$, since $p$ cannot simulate any action in $A^l\cup A^\mathit{bi}$. Then $q=\sum a_r q_r$ and we can
apply $(\mathsf{S}^r_p)$ to each summand in turn to get $0\< q$.

\item Let us consider $p=(\sum a_r p_r+ \sum a_l p_l + \sum a_b p_b)$, distinguishing the summands of $p$ which start with actions in either
$A^r$, $A^l$ or $A^\mathit{bi}$. We decompose $q$ in the same way to obtain $q=(\sum b_r q_r+ \sum b_l q_l + \sum b_b q_b)$. Then:

\begin{itemize}
\item For every $a_r$ there exists $b_r$, with $a_r=b_r$, such that $p_r\SimiCC q_r$ and, by induction hypothesis, $p_r\< q_r$. Then $\sum a_r p_r\< \sum b_r
q_r$. It could be the case that some summands of $\sum b_r q_r$ are never used to simulate any of the transitions of $p$, but then we can add
all those summand by using $(\mathsf{S}^r_p)$, to derive  $\sum a_r p_r\< \sum b_r q_r$.

\item For the summands $\sum a_l p_l$ and $\sum b_l q_l$ we can argue in exactly the same way, but starting with the righhand side and using
$(\mathsf{S}^{l}_p)$ instead of $(\mathsf{S}^r_p)$, to conclude now $\sum a_l p_l\< \sum b_l q_l$.

\item Finally, using standard arguments for bisimulation, we can establish a full correspondence between the summands $\sum a_b p_b$ and $\sum b_b
q_b$, having $a_b=b_b$ and $p_b\SimiCC q_b$, and by induction hypothesis we prove $\sum a_b p_b\< \sum b_b q_b$, thus concluding the proof. \qed
\end{itemize}
\end{itemize}
\end{proof}

\subsection{Conformance semantics}

Conformance simulation combines in a curious manner the features of both ordinary (covariant) and inverse (contravariant) simulation: the addition of new capabilities is always considered beneficial but, when an action is already offered, new ways to execute it are avoided since this leads to a more non-deterministic process.

To capture the first situation we need a variant of the axiom $(\mathsf{S})$ characterizing ordinary simulation:
\[(\mathsf{S}_{CS})\quad {I(p)\cap I(q)=\emptyset}\Longrightarrow {{p}\<{p+q}}.\]
For the latter, we instantiate the axiom $(\mathsf{S}^{-1})$ obtaining
\[(\mathsf{S}^{-1}_{CS})\quad {I(q)\subseteq I(p)}\Longrightarrow {{p+q}\<{p}},\]
which can be equivalently stated as

\[(\mathsf{S}^{-1}_{CS,p})\quad {{ap+aq}\<{ap}}.\]

There is, however, an important drawback: conformance simulation is not a precongruence because it is not always preserved by $+$. Indeed,
$\Conformance{0}{ab}$ and $\Conformance{ac}{ac}$, but not $\Conformance{ac}{ab+ac}$. Fortunately, to obtain a satisfactory algebraic treatment
of the conformance order it is enough to consider the weakest precongruence contained in it, as is done for weak bisimulation and the
corresponding observation congruence. Let us simply replace the axiom $(\mathsf{S}_{CS})$ by its guarded version
\[(\mathsf{S}_{CS,g})\quad {I(p)\cap I(q)=\emptyset}\Longrightarrow {{ap}\<{a(p+q)}}.\]

\begin{definition}
We define the conformance precongruence relation ${\PreConformance{p}{q}}$ by
\[{\PreConformance{p}{q}}\iff {(\Conformance{p}{q}\textrm{ and }{I(p)\supseteq I(q)})}.\]
\end{definition}

Note that the condition $I(p)\supseteq I(q)$ is not imposed recursively but just on the initial states of the processes, which corresponds to the fact that the (once) guarded axiom $(\mathsf{S}_{CS,g})$ becomes sound for the classical substitution calculus, in order to characterize the conformance precongruence $\PreConformance{}{}$.

\begin{proposition}
If the set of actions $A$ is infinite, then the precongruence relation $\PreConformance{}{}$ is the 
coarsest precongruence contained in $\Conformance{}{}$.
\end{proposition}

\begin{proof}
Obviously, we have ${\PreConformance{}{}}\subseteq{\Conformance{}{}}$. If there were a larger precongruence, there would exist $p$ and $q$ with
$\Conformance{p}{q}$ but $I(q)\not\subseteq I(p)$: then, taking $a\in I(q)\setminus I(p)$ and $b\in A$ such that $q\stackrel{a\cdot b}{\not\lra}$ we would have $\NConformance{ab+p}{ab+q}$ (since $\NConformance{ab}{q}$).

Finally, both the prefix operator and $+$ preserve $\PreConformance{}{}$:

\begin{itemize}
\item If $\PreConformance{p}{q}$, then $\PreConformance{ap}{aq}$ since $I(ap)=I(aq)=\{a\}$, and for $aq\stackrel{a}{\lra}q$ we have
$ap\stackrel{a}{\lra}p$ with $\PreConformance{p}{q}$.

\item If $\PreConformance{p}{q}$, then $\PreConformance{ap+r}{aq+r}$ since $I(ap+r)=I(aq+r)=I(r)\cup\{a\}$, and for $aq+r\stackrel{a}{\lra}q$ we
have $ap+r\stackrel{a}{\lra}p$ with $\PreConformance{p}{q}$ and, whenever $aq+r\stackrel{b}{\lra}r'$ with $r\stackrel{b}{\lra}r'$, we trivially
have $ap+r\stackrel{b}{\lra}r'$.\qed
\end{itemize}
\end{proof}

\begin{proposition}
The set of axioms $\mathcal{A}_{CS}=\{\mathsf{B}_1,\mathsf{B}_2,\mathsf{B}_3,\mathsf{B}_4,\mathsf{S}_{CS,g}, \mathsf{S}^{-1}_{CS,p}\}$ is complete for the conformance precongruence relation $\PreConformance{}{}$.
\end{proposition}

\begin{proof}
We show by induction on the depth of $p$ that, whenever $\PreConformance{p}{q}$ (resp. $\PreConformance{bp}{bq}$), we have $\mathcal{A}_{CS}\vdash p\< q$ (resp. $\mathcal{A}_{CS}\vdash bp\< bq$).

\begin{itemize}
\item If $\PreConformance{0}{q}$, then also $q=0$ and $0\<0$ using $(\mathsf{S}^{-1}_{CS,p})$.

\item If $\PreConformance{b0}{bq}$, then we can apply $(\mathsf{S}_{CS,g})$ with $p=0$.
\end{itemize}
Let us now consider $p=\sum_{a_i\in I(p)} a_i p_{ij}$ and $q=\sum_{a_i\in I(q)} a_i q_{ik}$.

\begin{itemize}
\item If $\PreConformance{p}{q}$ then $I(p)=I(q)$ and $\Conformance{p}{q}$, so for each $q_{ik}$ there is some $p_{ij}$ with
$\Conformance{p_{ij}}{q_{ik}}$ and therefore we can apply the second induction hypothesis to conclude that $a_ip_{ij}\< a_iq_{ik}$. It is
possible that some summands $p_{ij}$ will be paired with no $q_{ik}$ in the step above, but then we can apply the axiom
$(\mathsf{S}^{-1}_{CS,p})$ to them to conclude the proof.

\item Assume that $\PreConformance{bp}{bq}$. If $I(p)=I(q)$ then we also have $\PreConformance{p}{q}$ and this corresponds to the situation above. However, in this case we could have $I(p)\subsetneq I(q)$; then $q = q' + r$, with $r$ the summands $\sum_{a_i\in I(q)\setminus I(p)}a_i q_{ik}$, $I(p)=I(q')$, and $\PreConformance{p}{q'}$ and hence $p\< q'$. Now, we conclude the proof by applying the axiom $(\mathsf{S}_{CS,g})$ to $q'$ and $r$.\qed
\end{itemize}
\end{proof}

\section{Axiomatization of the new simulation equivalences}

Next we discuss the axiomatizability of the equivalences induced by covariant-contravariant and conformance simulations, obtaining a finite
axiomatization for the latter, and also for the first, but only when the set $A^\mathit{bi}$ of bivariant actions is empty. Instead, we also
present the impossibility result proving that covariant-contravariant simulation is not axiomatizable if we have $A^\mathit{bi}\neq\emptyset$
and $A^r\cup A^l\neq\emptyset$.

\subsection{Covariant-contravariant simulation}

Let us first consider the case in which $A^\mathit{bi}=\emptyset$. In order to axiomatize the equivalence $\equiv^{r,l}_{CC}$ induced by $(A^r,A^l)$-simulation we apply the general procedure introduced in \cite{Frutos-EscrigGP09,AcetoEtAl07,Frutos-EscrigGP09b}, based on the characterization

\[p\equiv_S p+q\Longleftrightarrow q\lesssim_S p\;.
\]

\noindent Thus we obtain:
\begin{itemize}
\item[ ] $(\mathsf{S1}^{r,l}_\equiv)\quad a_r(x+b_r y)= a_r(x+b_r y)+a_rx$.

\item[ ] $(\mathsf{S2}^{r,l}_\equiv)\quad a_r x= a_r x+ a_r(x+b_l y)$.
\end{itemize}
Obviously, the characterization above becomes unsound when contravariant prefixes appear because the pure contravariant simulation satisfies
\[q\equiv_S^{-1} p+q\Longleftrightarrow q\lesssim^{-1}_S p\;.\]

\noindent Therefore, we must reverse the inequalities above to obtain the adequate axioms for contravariant prefixes:
\begin{itemize}
\item[ ] $(\mathsf{S3}^{r,l}_\equiv)\quad a_lx= a_lx + a_l(x+b_ry)$.

\item[ ] $(\mathsf{S4}^{r,l}_\equiv)\quad a_l(x+b_ly)= a_l(x+b_ly)+ a_lx$.
\end{itemize}

Now we would expect the set of axioms
$\mathcal{A}_{CC}^\equiv=\{\mathsf{B}_1,\mathsf{B}_2,\mathsf{B}_3,\mathsf{B}_4,\mathsf{S1}^{r,l}_\equiv,\mathsf{S2}^{r,l}_\equiv,\mathsf{S3}^{r,l}_\equiv,$
$\mathsf{S4}^{r,l}_\equiv \}$ to axiomatize $(A^r,A^l)$-simulation equivalence. Certainly, all the axioms in this set are sound;  in order to
prove completeness in the absence of actions $A^\mathit{bi}$, we start by stating the following lemma that gives us two useful derived axioms.

\begin{lemma}\label{lemma-completeness}
The following equalities are derivable:
\[
\begin{array}{c@{\qquad}l}
\{\mathsf{S1}^{r,l}_\equiv, \mathsf{S2}^{r,l}_\equiv\}\vdash\ a_r(x+p_r)=a_r(x+p_r) + a_r(x+p_l)& (\mathsf{DS1}^{r,l}_\equiv)\\
\{\mathsf{S3}^{r,l}_\equiv, \mathsf{S4}^{r,l}_\equiv\}\vdash a_l(x+p_l)=a_l(x+p_l) + a_l(x+p_r)& (\mathsf{DS2}^{r,l}_\equiv)
\end{array}\]
where $p_r$ (resp. $p_l$) denotes any process prefixed by actions in $A^r$ (resp. $A^l$); more formally, $p_r=\sum_{i\in I}a^i_r p_i$ (resp. $p_l=\sum_{j\in J}a^j_l p_j$).
\end{lemma}

\begin{proof}
We only show the case of $(\mathsf{DS1}^{r,l}_\equiv)$. We start by proving that $a_r (x+p_r) = a_r (x+p_r) + a_r x$ by induction over the size $|I|$ of $I$.
\begin{itemize}
\item If $|I|=0$, the result is trivial.
\item If $|I|=1$, we immediately obtain the result by applying the axiom $(\mathsf{S1}^{r,l}_\equiv)$.
\item For $|I|>1$, we take $I={I'}\cup \{i\}$ with $|I'|=|I|-1$. Note that $a_r (x+p_r)=a_r ((x+p'_r)+a^i_r p_i)$ so that, applying axiom $(\mathsf{S1}^{r,l}_\equiv)$, we obtain
\[a_r (x+p_r)=a_r ((x+p'_r)+a^i_r p_i)+a_r (x+p'_r)= a_r (x+p_r)+ a_r (x+p'_r).
\]
Using the induction hypothesis with the term $a_r (x+p'_r)$ leads to
\[a_r (x+p_r)  = a_r (x+p_r) + a_r (x+p'_r)  + a_r x,\]
and, reusing the equality $a_r (x+p_r)+ a_r (x+p'_r)=a_r (x+p_r)$ above, we obtain
\begin{equation}\label{eq1}
a_r (x+p_r) = a_r (x+p_r) + a_r x
\end{equation}
as desired.
\end{itemize}

\noindent Now, we can analogously prove the
equality
\begin{equation}\label{eq2}
a_r x = a_r x + a_r (x+ p_l).
\end{equation}
Replacing $a_r x$ in equation~\ref{eq1} by the righthand side of
equation~\ref{eq2} produces
\[a_r (x+p_r) = a_r (x+p_r) + a_r x + a_r (x+ p_l)
\]
and, applying equation~\ref{eq1} again, we finally obtain
$(\mathsf{DS1}^{r,l}_\equiv)$:
\[a_r (x+p_r) = a_r (x+p_r) + a_r (x+ p_l).
\]
\qed
\end{proof}

For the main proof we have to adapt the classic technique for the completeness of the axiomatization of the plain simulation semantics ($p\lesssim_S q$ implies
$\mathcal{A}_S\vdash q=p+q$), taking into account the difference between covariant and contravariant actions. For technical reasons we need to
consider a ``free'' arbitrary term $r$.

\begin{proposition}\label{prop4}
If $p\lesssim_{CC} q$ then, for all processes $r$:
\[
\mathcal{A}_{CC}^{\equiv}\vdash a_r(q+r)=a_r(q+r) + a_r(p+r)
\]
and
\[
\mathcal{A}_{CC}^{\equiv}\vdash a_l(p+r)=a_l(p+r) + a_l(q+r).
\]
\end{proposition}

\begin{proof}
We proceed by induction on the depth of $p$. We start by decomposing both $p$ and $q$ as follows: $p=p_r+p_l$, $q=q_r+q_l$, where $p_r=\sum_{i\in I_{p_r}}a^i_{r} p_i$, $p_l=\sum_{i\in I_{p_l}}a^i_l
p_i$, $q_r=\sum_{i\in I_{q_r}}a^i_r q_i$ and $q_l=\sum_{i\in I_{q_l}}a^i_l q_i$. Then, it is clear that the depths of both $p_r$ and $p_l$ are less or equal than the depth of $p$ and besides we have ${p\lesssim_{CC}q}\Longleftrightarrow{{p_r\lesssim_{CC}q_r}\wedge{p_l\lesssim_{CC}q_l}}$.

Next, let us consider $p_r\lesssim_{CC} q_r$: this is an instance of the hypothesis of the statement to prove, which corresponds to the particular case in which $I(p)\cup
I(q)\subseteq A_r$. Then, we need to prove both
\[
\mathcal{A}_{CC}^{\equiv}\vdash a_r(q+r)=a_r(q+r) + a_r(p+r)
\]
and
\[
\mathcal{A}_{CC}^{\equiv}\vdash a_l(p+r)=a_l(p+r) + a_l(q+r).
\]

Let us consider in detail the second statement.
\begin{itemize}
\item If $p=0$, it follows that $\mathcal{A}_{CC}^{\equiv}\vdash a_lr=a_lr + a_l(q+r)$ by an application of the equation
$(\mathsf{DS2}^{r,l}_\equiv)$, with $p_l=0$, $x=r$, and $p_r=q$.

\item If $p=\sum_{i\in I}a_r^i p'_i$ and $q=\sum_{i\in J}a_r^i p'_i$, from $p\lesssim_{CC}q$ it follows, without loss of generality, that $I\subseteq J=I\cup J'$ and then we take $J=I\cup J'$ with $J'$ chosen such that $J'\cap I=\emptyset$, with $p'_i\lesssim_{CC}q_i'$ for all $i\in I$. Now,
by induction hypothesis, $\mathcal{A}_{CC}^{\equiv}\vdash a_r^i q'_i=a_r^i q'_i + a_r^i p'_i$. Next we obtain $\mathcal{A}_{CC}^{\equiv}\vdash \sum_{i\in I} a_r^i
q'_i=\sum_{i\in I} a_r^i q'_i + p$ and hence, by adding $\sum_{i\in J'} a_r^i q'_i$ to both sides, $\mathcal{A}_{CC}^{\equiv}\vdash q=q+p$, by congruence, we have $\mathcal{A}_{CC}^{\equiv}\vdash q+r=q+p+r$. Now, by applying $(\mathsf{DS2}^{r,l}_\equiv)$ with $x=p+r$, $p_l=0$,
and $p_r=q$, we obtain $\mathcal{A}_{CC}^{\equiv}\vdash a_l(p+r)=a_l(p+r)+a_l(p+r+q)$ which, combined with the previous equation, finally leads
to $\mathcal{A}_{CC}^{\equiv}\vdash a_l(p+r)=a_l(p+r)+a_l(q+r)$.
\end{itemize}

The first statement above is proved in a similar way, and the ones arising from $p_l\lesssim q_l$ can be dealt with analogously.

To conclude, we consider the general case $\CoCo{p}{q}$. By applying the results obtained above, starting from both $\CoCo{p_r}{q_r}$ and $\CoCo{p_l}{q_l}$, we have
\[
\mathcal{A}_{CC}^{\equiv}\vdash a_r(q_r+r)=a_r(q_r+r) + a_r(p_r+r)
\]
and
\[
\mathcal{A}_{CC}^{\equiv}\vdash a_r(q_l+r)=a_r(q_l+r) + a_r(p_l+r).
\]
In particular, making $r$ equal to $q_l+r'$ in the first equality:
\[ \mathcal{A}_{CC}^{\equiv}\vdash a_r(q_r+q_l+r')=a_r(q_r+q_l+r') + a_r(p_r+q_l+r').\]
(It is at this point that the ``free'' variable $r$ in the statement is needed, so as to be able to proceed by instantiating it in a suitable
manner). Now, instantiating $r$ with $p_r+r'$ in the second derived equation:
\[\mathcal{A}_{CC}^{\equiv}\vdash a_r(p_r+q_l+r')=a_r(p_r+q_l+r') + a_r(p_r+p_l+r').\]

If we now combine the last two equations we can obtain
\[ \mathcal{A}_{CC}^{\equiv}\vdash a_r(q_r+q_l+r')=a_r(q_r+q_l+r') + a_r(p_r+p_l+r'),\]
and, since $r'$ is arbitrary, we finally get
\[
\mathcal{A}_{CC}^{\equiv}\vdash a_r(q+r)=a_r(q+r) + a_r(p+r).
\]

We can proceed in a similar way for $a_l$, thus obtaining
\[
\mathcal{A}_{CC}^{\equiv}\vdash a_l(p+r)=a_l(p+r) + a_l(q+r).
\]

And this concludes the proof.\qed
\end{proof}

The main theorem is now at hand.

\begin{theorem}
Whenever $A=A^r\cup A^l$, the set of axioms
$\mathcal{A}_{CC}^\equiv=\{\mathsf{B}_1,\mathsf{B}_2,\mathsf{B}_3,\mathsf{B}_4,\mathsf{S1}^{r,l}_\equiv,\mathsf{S2}^{r,l}_\equiv,\mathsf{S3}^{r,l}_\equiv,$
$\mathsf{S4}^{r,l}_\equiv \}$ is complete for $(A^r,A^l)$-simulation equivalence.
\end{theorem}

\begin{proof}
Let $p\equiv_{CC}q$: we need to prove $\mathcal{A}_{CC}^\equiv\vdash p=q$. The proof will follow by induction on the depth of $p$.
\begin{itemize}
\item If $p=0$ we obviously have $q=0$.
\item Let $p=\sum_{i\in I}a_r^ip_r^i+\sum_{j\in J}a_l^jp_l^j$ and $q=\sum_{i\in I'}a_r^iq_r^i+\sum_{j\in J'}a_l^jq_l^j$. Then,
\begin{itemize}
\item for each $i\in I$, there exists some $i'\in I'$ with $a^i_r=a^{i'}_r$ and $p_r^i\s_{CC}q^{i'}_r$, and

\item for each $i'\in I'$ there exists some $i''\in I$ with $a^{i'}_r=a^{i''}_r$ and $q_r^{i'}\s_{CC}p^{i''}_r$.
\end{itemize}
Obviously, it could be the case that $i\neq i''$. Then, we could repeat the same argument with $i_1=i''$, and with $i_2=i_1''$, \dots, to obtain a sequence $(i, i_1, i_2,\dots)$. Since $|I|<\infty$, eventually we will find $i_m = i_n$ and,
hence,
\begin{itemize}
\item for each $i\in I$ we obtain $i'\in I'$ and $i''\in I$ such that $a^i_r=a^{i'}_r=a^{i''}_r$, $\CoCo{p^i_r}{q^{i'}_r}$ and
$p^{i''}_r\equiv_{CC}q^{i'}_r$.
\end{itemize}

Of course, we can repeat the same reasoning starting with $i'\in I'$ as well as for the contravariant summands in a dual way, to obtain the
following decompositions:
\[p=\sum_{i\in I}a^i_rp^i_r+\sum_{k\in K}a^k_rp^k_r+\sum_{k'\in K'}a^{k'}_lp^{k'}_l+\sum_{m\in M}a^m_lp^m_l
\]
\noindent and,
\[q=\sum_{i'\in I'}a^{i'}_rq^{i'}_r+\sum_{k\in K}a^k_rq^k_r+\sum_{k'\in K'}a^{k'}_lq^{k'}_l+\sum_{m'\in M'}a^{m'}_lq^{m'}_l
\]
\noindent where:
\begin{itemize}
\item for all $i\in I$, there exists $k\in K$ such that $a^i_r=a^k_r$ and $\CoCo{p^i_r}{p^k_r}$; and

\item for all $m\in M$, there exists $k'\in K'$ such that $a^m_l=a^{k'}_l$ and $\CoCo{p^{k'}_l}{p^m_l}$; and

\item for all $i'\in I'$, there exists $k\in K$ such that $a^{i'}_r=a^{k}_r$ and $\CoCo{q^{i'}_r}{q^k_r}$; and

\item for all $m'\in M'$, there exists $k'\in K'$ such that $a^{m'}_l=a^{k'}_l$ and $\CoCo{q^{k'}_l}{q^m_l}$; and

\item for all $k\in K$, $p^{k}_r\equiv_{CC}q^k_r$; and

\item for all $k'\in K'$, $p^{k'}_l\equiv_{CC}q^{k'}_l$.
\end{itemize}

Then we can apply the induction hypothesis to any pair $(p^k_r,q^k_r)$ and also to any pair $(p^{k'}_l,q^{k'}_l)$. To conclude the proof we only
need to apply Proposition~\ref{prop4}, taking $r=0$, to any such pairs $(p^i_r,p^k_r)$ and $(p^{k'}_l,p^{m}_l)$, and analogously for the
components of $q$.\qed
\end{itemize}
\end{proof}

The addition of bivariant actions (assuming that there are already other
actions present) changes the picture completely.
Now, it is no longer possible to axiomatize the equivalence.

\begin{theorem}
If $A^\mathit{bi}\neq\emptyset$ and $A^r\cup A^l\neq\emptyset$, then $(A^r,A^l)$-simulation equivalence is not finitely axiomatizable.
\end{theorem}

\begin{proof}
Let us take $a_{bi}\in A^\mathit{bi}$ and, without loss of generality, $a_r\in A^r$. We consider the two families of processes
\[p_n=a_r a_{bi}^n a_r 0\quad\textrm{and}\quad q_n=a_r a_{bi}^n a_r 0+a_r a_{bi}^n 0,
\]
where, as usual, we denote by $a^n_{bi}$ (with $n\geq 0$) the repeated application of the prefix operator $a_{bi}$ ($n$ times).

It is easy to check that $p_n\equiv_{CC} q_n$.
On the one hand, $\CoCo{p_n}{q_n}$ trivially; on the other hand, checking that $\CoCo{q_n}{p_n}$ simply amounts to checking that $\CoCo{0}{a_r}$. (However, note that taking $p^{-}_n=a_{bi}^n a_r 0$ and $q^{-}_n=a_{bi}^n a_r
0+a_{bi}^n 0$ does not lead to $p^{-}_n\equiv_{CC}q^{-}_n$; indeed, $p^{-}_n\not\s_{CC}q^{-}_n$ because if we start with the first $a_{bi}$ from the second summand of $q^{-}_n$ then $a_{bi}^{n-1} a_r 0\not\s_{CC}a_{bi}^{n-1} 0$.)
Now, for any finite axiomatization $\mathcal{A}$, let $n$ be bigger than the depth of any term appearing in $\mathcal{A}$;
we are going to show that if
$\mathcal{A}$ is sound for $\equiv_{CC}$ then we cannot have ${\mathcal{A}}\vdash{p_n=q_n}$.

We will show that if we start with $p_n$ and obtain a sequence of equivalent terms $p_n = p^{1}_n = p^{2}_n = \dots$, where each term
is obtained from the previous one by an application of a single axiom in $\mathcal{A}$, then no $p^{j}_n$ can be $q_n$.
If we apply an axiom to $p_n$ in a position different from its root, then we are transforming a subprocess $p'=a^m_{bi} a_r 0$,
with $m\leq n$, into some equivalent process $q\equiv_{CC}p'$. If we define $q\da m$ as the process obtained by ``pruning'' $q$ at
depth $m$, the result will be bisimilar to $a^m_{bi} 0$, since $q$ cannot execute any other action until it executes the prefix
$a_{bi}$ $m$ times and, moreover, it cannot stop in the meantime. In a similar way, from $q\equiv_{CC}p'$ we also infer that $q\da(m+1)\sim p'\da(m+1)$ and then the obtained $p^1_n$ satisfies  $p^1_n\da (n+2)\sim p_n$. The same argument can be applied starting from any $p^j_n$ such that $p^j_n\da (n+2)\sim p_n$, so that this invariant is preserved as long as there is no application of an axiom in $\mathcal{A}$ at the root of any $p^j_n$.

Therefore, the only possible way to break this invariant, that obviously is not satisfied by $q_n$, is to apply an axiom from $\mathcal{A}$ at the root of some $p^{j}_n$. In that case, the lefthand side of such an axiom would match several prefixes of the process $a_r a_{bi}^m0$ and then, following \cite{Groote90}, it is easy to see that the corresponding axiom has to be correct under bisimulation, too. As a consequence, the process $p^j_{n+1}$ resulting after the application of the axiom also satisfies $p^j_{n+1}\da(n+2)\sim p_n$. Therefore by repeated application of the axioms in $\mathcal{A}$ we will never reach a term such as $q_n$, thus concluding $\mathcal{A}\nvdash p_n=q_n$.\qed
\end{proof}

Note that the proof would remain valid even if we allowed conditional axioms whose conditions only observed the process locally, since the key fact
in the proof above is that in order to generate the choice at $q_n$ we need to ``see from the top'' that the two branches below, even if different from each other, can be joined to obtain a process equivalent to $p_n$.
But the branches cannot be joined bottom up, in a step by step fashion, since
 $p^{-}_n\not\equiv_{CC}q^{-}_n$.
Therefore, a conditional axiomatization whose conditions observe the processes locally would suffer the same problems as a purely equational one.

\subsection{Conformance simulation}

As before, we start by applying to the axioms characterizing $\lesssim^p_{CS}$ the general procedure presented in
\cite{Frutos-EscrigGP09,AcetoEtAl07,Frutos-EscrigGP09b}. In this case we obtain the following two axioms:

\begin{itemize}
\item[ ] $(\mathsf{S}^{CS}_\equiv)\quad {I(p)\cap I(q)=\emptyset}\Longrightarrow {ap=ap+a(p+q)}$.

\item[ ] $(\mathsf{S}^{-1,CS}_\equiv)\quad {I(q)\subseteq I(p)}\Longrightarrow {a(p+q)=a(p+q)+ap}$.
\end{itemize}

Note that we have used the contravariant version of the procedure because once we compare two processes offering the same set of actions the behavior of $\lesssim^p_{CS}$ is contravariant since we have
\[ap\gtrsim^p_{CS} ap+aq.\]

\noindent Therefore, we cannot apply the general results in \cite{Frutos-EscrigGP09,Frutos-EscrigGP09b} to prove the completeness of the proposed axiomatization. However, a beautiful variant of the classical proof for plain simulation will do the job.

\begin{theorem}
The set of axioms $\mathcal{A}_{CS}^\equiv=\{\mathsf{B}_1,\mathsf{B}_2,\mathsf{B}_3,\mathsf{B}_4,\mathsf{S}^{CS}_\equiv,\mathsf{S}^{-1,CS}_\equiv\}$ is
a complete axiomatization for the simulation equivalence $\equiv_{CS}$.
\end{theorem}

\begin{proof}
First note that $p\equiv_{CS}q$ implies $I(p)=I(q)$ and $p\equiv_{CS}^p q$, and therefore we can use either $\equiv_{CS}$ or $\equiv_{CS}^p$,
indistinctly. It is also routine to check the correctness of the axioms for $\equiv_{CS}$. To prove completeness, we show that
$\PreConformance{p}{q}$ implies $\mathcal{A}_{CS}^\equiv\vdash p=p+q$. Obviously, then we are done because  ${p\equiv_{CS}q}$ implies
$\PreConformance{p}{q}$ and $\PreConformance{q}{p}$.

We proceed by induction on the depth of $p$:

\begin{itemize}
\item $p=0$ implies $q=0$ trivially.

\item Let $\PreConformance{p}{q}$ with $p\stackrel{a}{\lra}$. Then we also have $q\stackrel{a}{\lra}$ and for all $q'$ with $q\stackrel{a}{\lra}q'$ there exists
$p\stackrel{a}{\lra}p'$ such that $\Conformance{p'}{q'}$. Note that we cannot conclude $\PreConformance{p'}{q'}$ since it is possible that
$I(p')\varsubsetneq I(q')$, but then we can write $q'=q''+r$ with $I(q'')=I(p')$ and $I(r)\cap I(q'')=\emptyset$. It is clear that $\PreConformance{p'}{q''}$, so that by induction hypothesis we obtain $\mathcal{A}^\equiv_{CS}\vdash p'=p'+q''$. Then, we have $\mathcal{A}_{CS}^\equiv\vdash ap'=a(p'+q'')$ and applying $(\mathsf{S}^{-1,CS}_\equiv)$, $\mathcal{A}_{CS}^\equiv\vdash ap'=a(p'+q'')+aq''$, and then $\mathcal{A}_{CS}^\equiv\vdash ap'=ap'+aq''$. Now, by applying $(\mathsf{S}^{,CS}_\equiv)$ we have $\mathcal{A}_{CS}^\equiv\vdash aq''=aq''+a(q''+r)$, to conclude that
$\mathcal{A}_{CS}^\equiv\vdash ap'=ap'+aq'$ and therefore  $\mathcal{A}_{CS}^\equiv\vdash p=p+q$. \qed
\end{itemize}
\end{proof}

Note that $(\mathsf{S}^{-1,CS}_\equiv)$ is the axiom characterizing the ready simulation equivalence, from which we conclude that
${\equiv_{RS}}\subseteq{\equiv_{CS}}$. Obviously, the reverse inclusion is false since $(\mathsf{S}^{CS}_\equiv)$ is not sound for
$\equiv_{RS}$. For instance, $ab=_{CS}ab+a(b+c)$, but $a(b+c)\not\lesssim_{RS} ab$. In fact, we also have $a(b+c)\not\lesssim_S ab$, proving that
${\equiv_{CS}}\nsubseteq{\equiv_{S}}$. In order to obtain $\equiv_{RS}$ from $\equiv_{CS}$ we should strengthen the definition of the latter by
considering ready conformance simulations defined as plain conformance simulations, but only allowing pairs of processes satisfying $I(p)=I(q)$.
If we denote by $\lesssim_{RCS}$ the generated preorder we have the following result.

\begin{proposition}
${\lesssim_{RCS}}={\lesssim_{RS}^{-1}}$, and therefore ${\equiv_{RCS}}={\equiv_{RS}}$ and ${\lesssim^{-1}_{RS}}\subseteq {\lesssim_{CS}}$.
\end{proposition}

Since $(\mathsf{S}^{-1,CS}_\equiv)$ is the axiom that defines ready simulation equivalence, it can be presented in an
equivalent way avoiding the condition and thus obtaining a pure algebraic axiom. However, it is not clear whether axiom
$(\mathsf{S}^{CS}_\equiv)$ allows such a finite pure algebraic presentation, and in fact the same happens with the axiom $(\mathsf{S}_{CS})$ in
the axiomatization of the conformance preorder. Hence, it could be the case that both the conformance preorder and the induced equivalence are
not finitely axiomatizable using pure equational axioms, as is the case for ready trace semantics.

\section{Conclusions}

We have continued with the study of covariant-contravariant simulation and conformance simulation semantics started in \cite{FabregasFP09,FabregasFP10} by considering the axiomatization of the preorders and equivalences that they define.

We have showed that the desired axiomatizations can be obtained from that of the plain simulation preorder, whose
completeness proof can be adapted in a simple, but elegant manner to obtain the completeness of the new axiomatizations. Also, by applying a suitable variation of our ``ready to preorder'' techniques
\cite{Frutos-EscrigGP09} we have obtained the axiomatizations of the corresponding conformance simulation equivalence. Surprisingly,  we also succeeded in axiomatizating the equivalence for covariant-contravariant simulations but only in the particular case where $A^\mathit{bi}=\emptyset$; otherwise, we proved that the covariant-contravariant simulation equivalence has turned out to be the second known example of a semantics whose defining preorder can be finitely axiomatized, but the induced equivalence cannot.
The first example of such a borderline situation can be found in \cite{ChenF08}.
It is curious to notice that although the two semantics are completely different
(the semantics here is quite simple since it is a plain semantics,
while the one in \cite{ChenF08} is much more complicated), and in our case it
is clear that the difficulties stem from the interference between bivariant
and monovariant actions, the structure of the considered ``counterexamples''
in both cases is essentially the same: there is a choice betweeen two quite
long branches which can be can joined into a single one,
but this should be done in a single step because the choice cannot be delayed
at all, even if the beginnings of the two branches are the same.
Therefore, in order to capture the equivalence, we would need an axiom
able to ``see'' the (too far away) ends of the two branches, but this is of
course impossible with a finite number of axioms
since the lengths of the branches in the counterexamples can be arbitrarily
long.

We expect our work on the subject to contribute to a better understanding of all the complex situations that arise when covariant and contravariant concepts coexist. This, for example, is the case in all the recent works on modal, input-output or interface formalisms, that try to clarify the relationships betwen specifications and implementations. In fact, it is our intention to continue with this line of research by trying to discover, and take benefit from all the connections between our work and those cited in this paper.


\end{document}